\documentclass[a4paper,onecolumn]{IEEEtran}
\IEEEoverridecommandlockouts

\usepackage{standalone}
\usepackage{amsmath,amssymb,lmodern,bm}
\usepackage[utf8]{inputenc}
\usepackage[T1]{fontenc}
\usepackage{adjustbox}
\usepackage{float}
\usepackage{pgfpages,graphicx,epstopdf}
\usepackage{tikz,pgfplots,fancyhdr,tcolorbox}
\usepackage{hyperref,mathtools,stackrel,array}
\usepackage{adjustbox}
\usepackage{pgfplots}
\usepackage{grffile}
\usepackage{tabularx}
\usepackage{booktabs}
\usepackage{totcount}
\usepackage{braket}
\usepackage{booktabs}
\usepackage{dsfont}
\usepackage{xcolor}
\usepackage[caption=false, font=footnotesize]{subfig}
\usetikzlibrary{arrows.meta, positioning}
\usetikzlibrary{external}
\allowdisplaybreaks
\def\BibTeX{{\rm B\kern-.05em{\sc i\kern-.025em b}\kern-.08em
    T\kern-.1667em\lower.7ex\hbox{E}\kern-.125emX}}

\newtheorem{theorem}{Theorem}
\newtheorem{lemma}{Lemma}
\newtheorem{proposition}{Proposition}
\newtheorem{corollary}{Corollary}

\newcommand{\supp}{\operatorname{supp}}

\newcommand{\Var}[1]{\textnormal{\textsf{Var}}\!\left({#1}\right)} 

\newcommand{\E}[2][]{\textnormal{\textsf{E}}_{#1}\!\left[#2\right]} 
\newcommand{\Prob}{\textnormal{Pr}} 
\DeclareMathOperator{\Tr}{Tr}

\title{Non-Asymptotic Bounds via the Nussbaum–Szkoła Mapping
for Quantum Hypothesis Testing}
\author{Jorge Lizarribar-Carrillo, Gonzalo Vazquez-Vilar,~\IEEEmembership{Member,~IEEE,} and Tobias Koch,~\IEEEmembership{Senior Member,~IEEE}
\thanks{The authors have received funding from the Spanish Ministerio de Ciencia, Innovación y Universidades under Grant PID2024-159557OB-C21 (MICIU/AEI/10.13039/501100011033 and ERDF/UE) and from the Comunidad de Madrid under Grant IDEA-CM (TEC-2024/COM-89). The material in this paper was presented in part at the 2066 IEEE International Symposium on Information Theory, Guangzhou, China, June--July 2026 [DOI:10.1109/ISIT62367.2026.11653873].}
\thanks{J.~Lizarribar-Carrillo, G.~Vazquez-Vilar and T.~Koch are with the Department of Signal Theory and Communications, Universidad Carlos III de Madrid, 28911 Leganés,  Spain, and also with the
Gregorio Marañón Health Research Institute, 28007 Madrid, Spain (emails: jlizarri@ing.uc3m.es, gonzalo.vazquez@uc3m.es, tkoch@ing.uc3m.es).}}

\begin{document}
\pagestyle{plain}

\maketitle

\begin{abstract}
Quantum hypothesis testing concerns the discrimination between quantum states. This paper introduces novel lower and upper bounds for asymmetric quantum hypothesis testing based on the Nussbaum–Szkoła mapping. Together, these bounds provide a unified recovery of both converse and achievability results across all major asymptotic regimes, including large-, moderate-, and small-deviations. Unlike existing bounds, which either rely on technically involved information-spectrum arguments or suffer from fixed prefactors and limited applicability in the non-asymptotic regime, the proposed bounds arise from a unified framework and enable, in some cases, the direct use of classical results. It is further demonstrated that these bounds provide accurate approximations to the optimal quantum error trade-off function at small blocklengths. Numerical comparisons with existing bounds, including those based on fidelity and information spectrum methods, highlight their improved tightness.
\end{abstract}


\section{Introduction}

Hypothesis testing is a central problem in information theory, concerned with distinguishing between two competing sources based on observed data. In the quantum setting, this task amounts to distinguishing between copies of two quantum states using measurements described by positive operator-valued measures (POVMs). The non-commutativity of quantum states gives rise to phenomena with no classical analogue, making the characterization of optimal error probabilities substantially more challenging.

A binary hypothesis test involves two types of error: a \mbox{type-I} error, which occurs when the null hypothesis is incorrectly rejected, and a type-II error, which occurs when it is incorrectly accepted. The fundamental trade-off between these errors is typically formalized by optimizing one error probability subject to a constraint on the other.
In the quantum setting, this trade-off admits an exact characterization via a semidefinite program (SDP). 
However, the dimension of the resulting SDP grows rapidly with the dimension of the quantum states.
This rapid scaling renders exact computations infeasible beyond very small systems, thereby motivating the development of analytically tractable and computationally efficient bounds that capture the essential behavior of the error probability.

Two converse bounds are particularly relevant in this context: 
The first, established by Nussbaum and Szkoła~\cite{nussbaum2009lower}, provides a fundamental link between quantum and classical hypothesis testing. This result bounds the optimal average error probability by one-half of the classical optimal error probability associated with the so-called Nussbaum–Szkoła distributions, which  are explicitly constructed from the spectral decompositions of the original quantum states $\rho$ and $\sigma$, effectively embedding the non-commutative geometry of the Hilbert space into a joint classical distribution. 
The second, proposed by Pereira \emph{et al.}~\cite{pereira2023analytical}, provides a lower bound in the asymmetric case based on the quantum fidelity $0 \leq F(\rho,\sigma) \leq 1$, which quantifies the overlap between the two quantum states.

Regarding achievability bounds, we highlight two relevant results. Audenaert \emph{et al.}~\cite{audenaert2012quantum} established an upper bound in terms of the quantum Chernoff overlap, a quantity defined through a geometric mean of the two quantum states, building upon earlier results by Hayashi~\cite{hayashi2006error}. Second, Cheng \emph{et al.}~\cite{cheng2026multiple} established an upper bound via the harmonic mean, which can further be used to derive the achievability counterpart to the Nussbaum--Szkoła bound; namely, it proves that the optimal average error probability is upper-bounded by twice the classical optimal error probability under the Nussbaum--Szkoła distributions.

The asymptotic analysis of quantum hypothesis testing considers the discrimination of multiple copies of quantum states, described by the tensor-product states $\rho^{\otimes n}$ and $\sigma^{\otimes n}$. For large numbers of copies~$n$, asymptotic methods characterize the type-II error probability $\beta_n$ under a prescribed scaling of the type-I error $\alpha_n$. The \textit{small-deviations regime} covers the asymptotic analysis of $\beta_n$ when $\alpha_n$ is fixed. The quantum Stein lemma establishes the error exponent in this setting \cite{hiaipetz1991, ogawahayashi2004}, thereby providing the leading term of order $n$ in the expansion of $-\log \beta_n$. A second-order analysis refines this result by characterizing the correction term of order $\sqrt{n}$. In the quantum setting, two pioneering works approach this problem independently from different perspectives: Li~\cite{li2014second} employs elementary algebraic techniques and probability bounds derived from the Nussbaum–Szkoła mapping to characterize the second-order term, while Tomamichel and Hayashi~\cite{tomamichel2013hierarchy} use one-shot entropic quantities---such as the information spectrum relative entropy---to obtain the same result.

When the type-I error decays subexponentially, the problem falls in the \textit{moderate-deviations regime}. The exponent of the type-II error was analyzed asymptotically in~\cite{chubb2017moderate, cheng2017moderate}. In \cite{chubb2017moderate}, Chubb \emph{et al.} used a combination of relative entropy inequalities~\cite{tomamichel2013hierarchy} and a Berry–Esseen–type result due to Rozovsky~\cite{rozovsky2002estimate}. In~\cite{cheng2017moderate}, Cheng and Hsieh derived the achievability by using a martingale inequality due to Sason~\cite{sason2012moderate}, originally developed for classical moderate deviations, while their converse follows directly from a sharp converse Hoeffding bound.

When both errors decay exponentially, the problem falls in the \emph{large-deviations regime}, with performance described by the error exponents $A = -\lim_{n\to\infty}\tfrac{1}{n}\log \alpha_n$ and $B = -\lim_{n\to\infty}\tfrac{1}{n}\log \beta_n$. In the quantum setting, the converse part of this analysis was established by Nussbaum and Szkoła~\cite{nussbaum2009lower}, while the achievability was proven by Audenaert \emph{et al.}~\cite{audenaert2007discriminating}, together yielding the quantum generalization of the classical Chernoff exponent. Based on these results, the Hoeffding bound in the asymmetric setting was established by Hayashi and Nagaoka~\cite{hayashi2006error,nagaoka2006converse}.

In this work, we establish bounds for asymmetric quantum hypothesis testing based on the Nussbaum--Szkoła mapping. First, we introduce a converse bound in terms of the classical error trade-off of the corresponding Nussbaum--Szkoła distributions, parameterized by a prefactor $s \in [0,1]$. For a testing problem between pure and mixed states, we complement this result with an analogous upper bound where the type-I and type-II errors of the corresponding Nussbaum–Szkoła test are weighted by $1/s_1$ and $1/s_2$, respectively, where $(s_1,s_2)$ belong to a certain region of $[0,1]^2$. For $s_1=s$ and $s_2=1-s$, this bound recovers an achievability bound recently established by Cheng et al. [3] for a test between arbitrary states. Together, these bounds provide a unified and conceptually simple framework that recovers known results across the main asymptotic regimes of quantum hypothesis testing: large-, moderate-, and small-deviations.
This unification follows from a single expression by selecting appropriate sequences of the parameter $s \equiv s_n$. In certain regimes, it further enables a direct transfer of classical hypothesis testing results to the quantum setting.
We further show that our converse bound is asymptotically tight for pure, asymptotically orthogonal states, implying that the associated multiplicative constants cannot be improved in general.

The rest of the paper is organized as follows. Section~\ref{sec:HT} introduces the binary hypothesis testing problem and revisit basic results regarding the quantum error trade-off function. Section~\ref{sec:NS} discusses the Nussbaum-Szkoła mapping. Sections~\ref{sec:LB} and \ref{sec:UB} present our lower bound and upper bound on the error probability. Section~\ref{sec:asymp} analyzes these bounds asymptotically in the large-, moderate-, and small-deviations regimes and demonstrates that they are asymptotically tight in these regimes. Section~\ref{sec:discussion} concludes the paper with a comparison of our lower and upper bounds with previous bounds obtained in the literature in the non-asymptotic setting.

\subsection{Notation}
Throughout this paper, finite-dimensional complex Hilbert spaces are denoted by capital calligraphic letters such as $\mathcal{H}$, with dimension $d \triangleq \dim(\mathcal{H})$. The identity operator on $\mathcal{H}$ is denoted by $I$. We denote the space of bounded linear operators on $\mathcal{H}$ by $\mathcal{B}(\mathcal{H})$, and the set of density operators (self-adjoint, positive semidefinite operators with unit trace) by $\mathcal{D}(\mathcal{H})$. For self-adjoint operators $A, B \in \mathcal{B}(\mathcal{H})$, $A \succeq B$ (or $B \preceq A$) indicates that $A - B$ is positive semidefinite, and $A \succ 0$ denotes positive definiteness. The trace of an operator $A$ is written as $\Tr[A]$. For a self-adjoint operator $A$, $\{A \succeq 0\}$ (and respectively $\{A \succ 0\}$) denotes the orthogonal projection onto the subspace spanned by eigenvectors corresponding to non-negative (respectively strictly positive) eigenvalues.

\section{Hypothesis testing}
\label{sec:HT}

We study the problem of discriminating between two quantum states. Specifically, let us consider the density operators $\rho$ and $\sigma$,  acting on some finite-dimensional complex Hilbert space~$\mathcal{H}$ with dimension $d$, and define the hypotheses 
\begin{equation}
 \mathrm{H}_0\colon \rho,\qquad
 \mathrm{H}_1\colon \sigma.
\end{equation}
In this binary setting, we distinguish between two error types:
\begin{itemize}
    \item The type-I error is the error of accepting $\mathrm{H}_1$ when the true state is the null hypothesis $\mathrm{H}_0\colon \rho$.
    \item The type-II error is the error of accepting $\mathrm{H}_0$ when the true system state is the alternative hypothesis $\mathrm{H}_1\colon \sigma$.
\end{itemize}
A binary test is defined by a positive self-adjoint operator $\Pi$ acting on~$\mathcal{H}$ such that $0\preceq \Pi \preceq I$. For a test $\Pi$ associated with $\mathrm{H}_1$, let $\bar\Pi \triangleq I-\Pi$. The type-I and type-II error probabilities are, respectively,
\begin{equation}
  \alpha(\Pi) = \Tr[\Pi \rho], \qquad \beta(\Pi) = \Tr[\bar\Pi \sigma] = 1 - \Tr[\Pi \sigma].
\end{equation}
The two error probabilities cannot be made arbitrarily small at the same time.
The best achievable trade-off between these probabilities corresponds to the Pareto-optimal boundary
\begin{equation}
  \beta_{\alpha}(\rho,\sigma) = \inf_{\Pi: \alpha(\Pi) \leq \alpha} \beta(\Pi).
  \label{eqn:beta-alpha-def}
\end{equation}
Let $A,B$ be positive semidefinite operators, define
\begin{IEEEeqnarray}{lCl}
    \Tr[\min\{A,B\}] &\triangleq& \min_{0 \preceq \Pi \preceq \mathds{1}} \Tr[A \Pi] + \Tr[B(\mathds{1} - \Pi)] \notag \\
    & = &\Tr[A\{A-B \preceq 0\}] + \Tr[B\{A-B \succ 0\}],\label{eq:trmin}
\end{IEEEeqnarray}
where the expressions in braces denote the spectral projections onto the indicated eigenspaces. The quantum error trade-off function admits the following variational formulation given in the following lemma:
\begin{lemma}[\hspace{-0.01mm}{\cite[Lemma~2]{vazquez2016multiple}}]\label{lemma:1}
Let  $\rho$ and $\sigma$ be two density operators. Then,
\begin{IEEEeqnarray}{lCl}
\beta_{\alpha}(\rho,  \sigma) &=& \sup_{t\geq 0} \left\{\Tr[\min\{t\rho,\sigma\}] - t\alpha \right\} \notag \\
&=&\sup_{t\geq 0} \left\{\Tr[\sigma \Pi_t] + t \left( \Tr[\rho\bar{\Pi}_t] - \alpha\right) \right\}, \label{eq:var_min}
\end{IEEEeqnarray}
where $\Pi_t \triangleq \{ t\rho - \sigma \succeq 0\}$  and $\bar{\Pi}_t = I - \Pi_t$. 
\end{lemma}

A direct consequence of this lemma is that, for the discrete classical distributions $P$ and $Q$, we have
\begin{equation}
\beta_{\alpha}(P,  Q) = \sup_{t\geq 0} \left\{ \sum_{x \in \mathcal{X}} \min\left\{tP(x), Q(x)\right\} - t\alpha\right\}. \label{eq:var_class}
\end{equation}
Indeed, \eqref{eq:var_class} follows directly from \eqref{eq:var_min} in the special case where $\rho$ and $\sigma$ commute and are simultaneously diagonalizable in a common basis with diagonal entries given by $P$ and $Q$, respectively. In this commuting regime, $\Pi_t$ reduces to the indicator projection onto the set $\{x \in \mathcal{X} \mid t P(x) \ge Q(x)\}$, simplifying the trace to a sum over $\min(t P(x), Q(x))$ for all $x \in \mathcal{X}$.

\section{The Nussbaum-Szkoła Mapping}
\label{sec:NS}

The Nussbaum--Szkoła mapping~\cite{nussbaum2009lower} provides a fundamental bridge connecting quantum states to classical probability distributions. Let $\rho$ and $\sigma$ be density operators acting on a finite-dimensional Hilbert space $\mathcal{H}$ of dimension $d$, with spectral decompositions given by
 \begin{equation}
\rho = \sum_{i=1}^d \lambda_i \ket{x_i} \bra{x_i}, \quad \sigma = \sum_{j=1}^d \mu_j\ket{y_j}\bra{y_j}, \label{eq:spec}
\end{equation}
where $\{\ket{x_i}\}_{i=1}^d$ and $\{\ket{y_j}\}_{j=1}^d$ are orthonormal bases of $\mathcal{H}$, and $\bm{\lambda} = (\lambda_1, \dots, \lambda_d)$ and $\bm{\mu} = (\mu_1, \dots, \mu_d)$ are the respective eigenvalue spectra. The Nussbaum-Szkoła mapping distributions are then defined by the overlaps between the eigenstates of $\rho, \sigma$ and their associated eigenvalues
\begin{equation}
P_{ij} \triangleq \lambda_i \lvert\braket{x_i | y_j}\rvert^2, \quad Q_{ij} \triangleq \mu_j\lvert\braket{x_i|y_j}\rvert^2. \label{eq:Ns-mapping}
\end{equation}
The Nussbaum–Szkoła distributions $P$ and $Q$ are thus defined over a discrete alphabet of cardinality $d^2$.
Although they do not generally correspond to the outcome statistics of a common quantum measurement, they provide a classical probabilistic representation of the spectral structure of $\rho$ and $\sigma$. In particular, their classical Chernoff overlap coincides exactly with the quantum Chernoff overlap. Therefore,
\begin{equation}
    \kappa(p) \triangleq \log \operatorname{Tr}\!\left[\rho^p \sigma^{1-p}\right] = \log \sum_{ij} P_{ij}^p Q_{ij}^{1-p},
\end{equation}
which implies the equality of the relative entropy, Petz--R\'enyi relative entropy, and higher-order cumulants, since
\begin{IEEEeqnarray}{lCcCl}
    D_p(\rho\|\sigma) &\triangleq & \frac{\kappa(p)}{p - 1} & = & D_p(P\|Q), \quad \forall p \in (0,1) \cup (1,\infty), \label{eq:D_p}\\
    D(\rho\|\sigma) &\triangleq & \left.\frac{d}{dp}\kappa(p)\right|_{p=1} & = & D(P\|Q), \label{eq:in} \\
    V(\rho \| \sigma) &\triangleq & \left.\frac{d^2}{dp^2}\kappa(p)\right|_{p=1} & = & V(P\|Q), \label{eq:relvariance}\\
    T(\rho\|\sigma) &\triangleq & \left.\frac{d^3}{dp^3}\kappa(p)\right|_{p=1} & = & T(P\|Q). \label{eq:fin}
\end{IEEEeqnarray}
Moreover, the mapping preserves tensor-product structures: for product states \(\rho^{\otimes n}\) and \(\sigma^{\otimes n}\), the corresponding Nussbaum--Szkoła distributions decompose into the i.i.d.\ product distributions \(P^{\otimes n}\) and \(Q^{\otimes n}\). These properties were used in the literature to establish several asymptotic results in quantum hypothesis testing (see, e.g.,  \cite{chubb2017moderate, nussbaum2006lower}). Intuitively, relating the quantum error probabilities to those of the corresponding classical distributions enables the direct application of classical limit theorems---including the central limit theorem, moderate deviation principles, and large deviation theory---to derive sharp converse and achievability bounds in quantum hypothesis testing.

\section{Error Probability Lower Bound}
\label{sec:LB}

\subsection{The Lower Bound}
The next theorem relates the quantum and classical hypothesis testing trade-off functions of \eqref{eq:spec} and \eqref{eq:Ns-mapping}, respectively.
Specifically, we show that the Pareto-optimal boundary of the quantum test  \eqref{eq:spec} is lower-bounded by the Pareto-optimal boundary of the test \eqref{eq:Ns-mapping}, with the type-I and type-II error probabilities weighted by the respective factors $1-s$ and $s$, where $s$ is an arbitrary parameter satisfying $0 \leq s < 1$.

\begin{theorem}\label{thm:low}
For a binary quantum hypothesis test between states $\rho$ and $\sigma$ \eqref{eq:spec} with Nussbaum-Szkoła mapping distributions $P$ and $Q$ \eqref{eq:Ns-mapping}, we have that
\begin{IEEEeqnarray}{lCl}
\beta_{\alpha}(\rho, \sigma)  &\geq & s \beta_{\frac{1}{1-s}\alpha}(P, Q), \label{eq:thmlow}
\end{IEEEeqnarray}
for every $0 \leq s < 1$ and $0 \leq \alpha \leq 1-s$.
\end{theorem}
\begin{IEEEproof} Using  that $\sum_{i=1}^d\!\ket{x_i}\bra{x_i}$ and $\sum_{j=1}^{d}\ket{y_j}\bra{y_j}$ are equal to the identity matrix, the idempotency of $\Pi_t$ and $\bar{\Pi}_t$, and the cyclic property of the trace, we rewrite the trace terms in Lemma~\ref{lemma:1} as
\begin{equation}
    \Tr[\sigma \Pi_t] = \sum_{ij} \mu_j \lvert \braket{x_i | \Pi_t | y_j}\rvert^2, \qquad
    \Tr[\rho \bar{\Pi}_t] = \sum_{ij} \lambda_i \lvert \braket{x_i | \bar{\Pi}_t | y_j}\rvert^2. \label{eq:term21} 
\end{equation} 
Substituting the identities \eqref{eq:term21} into the variational formulation of $\beta_\alpha(\rho,  \sigma)$ \eqref{eq:var_min}, we obtain for a fixed $t \geq 0$
\begin{equation}
\beta_\alpha(\rho,  \sigma) \geq \sum_{ij} \mu_j \lvert \braket{x_i | \Pi_t | y_j}\rvert^2 + \sum_{ij} t\lambda_i \lvert \braket{x_i | \bar{\Pi}_t | y_j}\rvert^2 - t\alpha.
\end{equation}
Let $t = \delta t'$ for $\delta \geq 0$. After joining both sums, and lower-bounding $t'\lambda_i$ and $\mu_j$ by $\min(t'\lambda_i, \mu_j)$, this yields
\begin{equation}
\beta_\alpha(\rho,  \sigma) \geq \sum_{ij} \min(t'\lambda_i, \mu_j)\Bigl( \lvert \braket{x_i | \Pi_t | y_j}\rvert^2  +  \delta \lvert \braket{x_i | \bar{\Pi}_t | y_j}\rvert^2 \Bigr) - \delta t'\alpha. \label{eq:ineqdelta}
\end{equation}
Next, for an arbitrary $0\leq s < 1$, define the vectors
\begin{equation}
u \triangleq \left[\sqrt{s},  \sqrt{1-s} \right], \qquad
v \triangleq \left[\braket{x_i | \Pi_t | y_j},  \sqrt{\delta} \braket{x_i | \bar{\Pi}_t | y_j} \right].\label{eq:vecs}
\end{equation}
The Cauchy-Schwarz inequality $\lvert\braket{u,v}\rvert^2 \leq \braket{u,u} \braket{v,v}$ yields
\begin{equation}
    \left|\sqrt{s}\braket{x_i | \Pi_t | y_j} +  \sqrt{1-s}\sqrt{\delta} \braket{x_i | \bar{\Pi}_t | y_j} \right|^2 
    \leq \lvert\braket{x_i | \Pi_t | y_j}\rvert^2 +  \delta \lvert \braket{x_i | \bar{\Pi}_t | y_j}\rvert^2. \label{eq:cauchyschwarz}
\end{equation}
By choosing $\delta = s/(1 - s)$ in \eqref{eq:cauchyschwarz}, and using that $\Pi_t + \bar{\Pi}_t = I$, we obtain
\begin{equation}
    s \lvert\braket{x_i | y_j} \rvert^2 \leq\lvert\braket{x_i | \Pi_t | y_j}\rvert^2 + \frac{s}{1-s} \lvert \braket{x_i | \bar{\Pi}_t | y_j}\rvert^2. \label{eq:cauchyschwarz2}
\end{equation}
This particular choice of $\delta$ enforces $\delta \geq 0$ for $0 \leq s < 1$.
Applying \eqref{eq:cauchyschwarz2} to \eqref{eq:ineqdelta} then yields
\begin{IEEEeqnarray}{lCl}
\beta_\alpha(\rho,  \sigma) &\geq & s\Biggl(\sum_{ij} \min(t'\lambda_i, \mu_j) \lvert\braket{x_i | y_j}\rvert^2  - \frac{1}{1-s}t'\alpha \biggr)\\
&= &s\Biggl(\sum_{ij} \min(t'P_{ij}, Q_{ij}) - \frac{1}{1-s}t'\alpha \biggr). \label{eq:ineqsclass}
\end{IEEEeqnarray}
By \eqref{eq:var_class}, taking the supremum over $t'$ in \eqref{eq:ineqsclass} completes the proof.
\end{IEEEproof}
\subsection{Pure-State Discrimination}
The bound in Theorem~\ref{thm:low} is asymptotically tight for a test between pure states that are asymptotically orthogonal. Consider the discrimination between the pure states
\begin{equation}
 \mathrm{H}_0\colon \rho = \ket{x_1}\bra{x_1}, \qquad
 \mathrm{H}_1 \colon \sigma = \ket{y_1}\bra{y_1},\label{eq:HT-pure-H0-H1}
\end{equation}
with overlap parameter
\begin{equation}
    w_1 \triangleq \left|\braket{x_1|y_1}\right|^2 \in (0,1).
\end{equation}
\subsubsection{Quantum test}
For the pure quantum states $\rho, \sigma$ given in \eqref{eq:HT-pure-H0-H1}, the quantum error trade-off is characterized in the following lemma.

\begin{lemma}[\hspace{-0.01mm}{Pereira et al.~\cite{pereira2023analytical}}]\label{lem:pure-tradeoff}
Let $\rho, \sigma$ be pure states. Then, the optimal quantum error trade-off is given by
\begin{equation}\label{eq:pure_betaq}
\beta_{\alpha}(\rho, \sigma) = \alpha - 2\alpha w_1 + w_1 - 2 \sqrt{(1-\alpha)\alpha(1-w_1)w_1},
\end{equation}
for $0\leq \alpha \leq w_1$, where $a$ is defined in~\eqref{eq:def_a}.
\end{lemma}
\begin{IEEEproof}
The lemma corresponds to \cite[Eq. (10)]{pereira2023analytical}. We include an alternative proof here for completeness. When $\rho, \sigma$ are pure, $\Tr[\min\{t\rho, \sigma\}]$ admits the closed-form expression \cite[Eq. (2.34)]{Helstrom1976}
\begin{equation}\label{eq:minpure}
\Tr[\min\{t\rho, \sigma\}] = \frac{1}{2}\left(1 + t - \sqrt{(t+1)^2 - 4tw_1}\right).
\end{equation}
Invoking Lemma~\ref{lemma:1} alongside \eqref{eq:minpure}, we obtain
\begin{equation}\label{eq:beta_t}
\beta_{\alpha}(\rho, \sigma) = \sup_{t} \left\{ \frac{1}{2} \left(1 + t - \sqrt{(t+1)^2 - 4tw_1}\right) - t\alpha \right\}.
\end{equation}
Since the objective function in \eqref{eq:beta_t} is concave in $t$ for $0 \leq \alpha \leq a$, differentiating with respect to $t$ and setting the derivative to zero yields the maximizing value
\begin{equation}\label{eq:tstar}
t^{\star} = 2w_1 - 1 + \frac{(1-2\alpha)\sqrt{w_1(1-w_1)}}{\sqrt{\alpha(1-\alpha)}}.
\end{equation}
Substituting \eqref{eq:tstar} into \eqref{eq:beta_t} yields \eqref{eq:pure_betaq} after rearranging terms.
\end{IEEEproof}

Expanding \eqref{eq:pure_betaq} as $w_1\to 0$, we obtain
\begin{equation}
   \beta_\alpha(\rho, \sigma) = \alpha - 2\sqrt{\alpha w_1 } + a + o(w_1), \quad 0 \leq \alpha \leq w_1,\label{eq:taylor}
\end{equation}
where $o(w_1)$ summarizes terms that vanish faster than $w_1$ as \mbox{$w_1\to 0$}. 
\subsubsection{Classical test}
For the above quantum states $\rho$ and $\sigma$,
the Nussbaum-Szkoła mapping from \eqref{eq:Ns-mapping} becomes
\begin{equation}
 P_{ij} = \begin{cases}
               \lvert\braket{x_1|y_j}\rvert^2 , \:\:\: i=1,\\
               0, \:\:\:\:\:\:\quad\text{otherwise},
           \end{cases} \qquad
 Q_{ij} = \begin{cases}
               \lvert\braket{x_i|y_1}\rvert^2 , \:\:\: j=1,\\
               0,\:\:\:\:\:\:\quad\text{otherwise}.
           \end{cases}
\end{equation}
The distributions $P$ and $Q$ exhibit non-overlapping supports,
except in the singular case $(i,j)=(1,1)$ where 
\begin{equation}
 P_{11} = Q_{11} = w_1. \label{eq:def_a}
\end{equation}
The optimal classical test for this problem decides on the correct hypothesis with no error, except when $(i,j)=(1,1)$. In this case, the optimal test can select between $\mathrm{H}_0$ and $\mathrm{H}_1$ at random with some probabilities $p$ and $1-p$, thus incurring an error with probabilities
\begin{equation}
  \alpha = (1-p) P_{11} = (1-p) w_1,\qquad
  \beta = p Q_{11} = p w_1.
\end{equation}
After solving for $p$ in $\alpha$ and substituting into $\beta$, we obtain the error trade-off
\begin{equation}
  \beta_{\alpha}(P,Q) = w_1 - \alpha, \quad 0 \leq \alpha \leq w_1.\label{eq:trade-off-pure-c}
\end{equation}
Therefore, using \eqref{eq:trade-off-pure-c}, the lower bound of Theorem~\ref{thm:low} can be written as
\begin{equation}
    \beta_{\alpha}(\rho,\sigma) \geq s\beta_{\frac{\alpha}{1-s}}(P,Q) = sw_1 - \frac{s}{1-s}\alpha, \quad 0\leq s < 1,
\end{equation}
which, for $s=1-\sqrt{\frac{\alpha}{w_1}}$, becomes
\begin{equation}
\beta_{\alpha}(\rho,\sigma) \geq \alpha - 2\sqrt{\alpha w_1} + w_1.\label{eq:pure-low-ab}
\end{equation}
Comparing \eqref{eq:taylor} with \eqref{eq:pure-low-ab}, we observe that the quantum error trade-off in \eqref{eq:taylor} coincides with the lower bound in \eqref{eq:pure-low-ab} up to terms of order $o(a)$.
For a test between the $n$-th tensor product states $\rho^{\otimes n},  \sigma^{\otimes n}$, we have that $a= \lvert \braket{x_1 | y_1} \rvert^{2n}$. So the assumption $\lvert \braket{x_1 | y_1} \rvert^2 < 1$ implies that $a \to 0$ as $n \to \infty$. This proves that, for distinct pure states, the bound is asymptotically tight as $n\to\infty$. 
\begin{figure}[tb]
    \centering
    \includegraphics[width=0.48\linewidth]{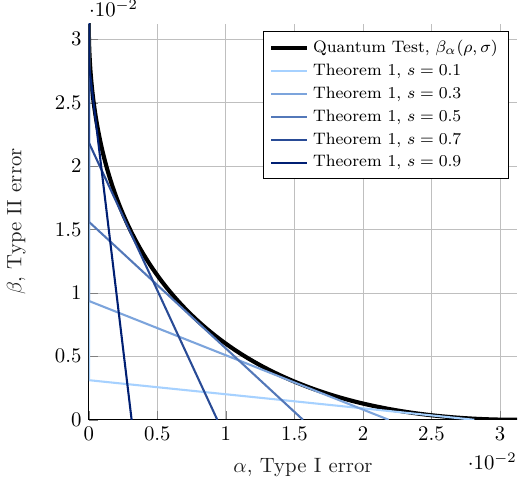}
    \caption{Quantum error trade-off compared to the lower bound from Theorem~\ref{thm:low} for a hypothesis test between the states ${\ket{0}\!\bra{0}}^{\otimes 5}$ and ${\ket{+}\!\bra{+}}^{\otimes 5}$.}
    \label{fig:pure}
\end{figure}

As an example, we consider the $n$-th tensor product of the states  $\rho=\ket{0}\!\bra{0}$ and $\sigma = \ket{+}\!\bra{+}$. Fig.~\ref{fig:pure} shows the quantum error trade-off $\beta_{\alpha}$ for the hypothesis test between $\rho^{\otimes n}$ and $\sigma^{\otimes n}$ when $n=5$, and compares it to the lower bound of Theorem~\ref{thm:low} for different values of $s$. Observe that the upper envelope of the resulting bounds provides an accurate characterization of the quantum error trade-off when both states are pure and approximately orthogonal (in this example $w_1 = 2^{-5}\approx 0.03$).
\section{Error Probability Upper Bound}
\label{sec:UB}
In this section, we first present an upper bound on $\beta_{\alpha} (\rho, \sigma)$, recently presented by Cheng \emph{et al.}~\cite[Prop.~6.3]{cheng2026multiple} for arbitrary quantum states $\rho$ and $\sigma$, which relates the quantum error trade-off to that of the corresponding Nussbaum--Szkoła distributions.
Then, in Theorem~\ref{thm_pure} we show that this result can be further tightened in the particular case in which one of the states is pure.
\begin{theorem}[\hspace{-0.01mm}{Cheng et al.~\cite[Prop.~6.3]{cheng2026multiple}}]\label{thm_general} Let $\rho, \sigma$ be arbitrary quantum states. Then
\begin{equation}
    \beta_{\alpha} (\rho, \sigma) \leq \frac{1}{s}\beta_{(1-s)\alpha}(P, Q),
\end{equation}
where $0 < s < 1$.
\end{theorem}
\begin{IEEEproof}
    See \cite[Th.~4.2, Prop.~6.3]{cheng2026multiple}.
\end{IEEEproof}

Together with Theorem~\ref{thm:low}, this result establishes complementary lower and upper bounds on the quantum hypothesis-testing quantity $\beta_\alpha(\rho,\sigma)$ in terms of its classical counterpart $\beta_\alpha(P,Q)$ associated with the Nussbaum--Szkoła mapping, with both bounds involving a multiplicative factor and a rescaling of the parameter $\alpha$.
However, in contrast to the lower bound from Theorem~\ref{thm:low}, the upper bound in Theorem~\ref{thm_general} can be tightened in the case in
which one of the two alternatives is pure. To this end, we first establish the following lemma.
\begin{lemma}\label{lemma3}
Let $A \succeq 0$ and $B \succeq 0$ have spectral decompositions
\begin{equation}
A=\sum_i a_i\ket{x_i}\bra{x_i},
\qquad
B=b_1\ket{y_1}\bra{y_1}.
\end{equation}
Thus, $B$ has rank $1$. Further define the overlap coefficients
\begin{equation}
w_i\triangleq \left|\braket{x_i|y_1}\right|^2. \label{eq:weights_i}
\end{equation}
Then,
\begin{equation}
\Tr[\min\{A,B\}]
\leq
\sum_i
\min\left(\frac{a_i}{s_1},\frac{b_1}{s_2}\right)w_i, \label{eq:ineqnspq}
\end{equation}
for every $(s_1, s_2) \in \mathcal{R}$, where 
\begin{equation}
\mathcal{R} \triangleq \left\{ (s_1, s_2) \in (0, 1]^2 \;\middle|\; s_1 \leq f(s_2) \right\} \label{eq:region}
\end{equation}
and 
\begin{equation}
    f(s_2) \triangleq \begin{cases} 1, & 0 < s_2 \le \frac{1}{2} \\[4pt] 4s_2(1-s_2), & \frac{1}{2} < s_2 < 1.\end{cases}
\end{equation}
\end{lemma}
\begin{IEEEproof}
    See Appendix~\ref{appendixA}.
\end{IEEEproof}

\begin{figure}[tb]
    \centering
    \includegraphics[width=0.48\linewidth]{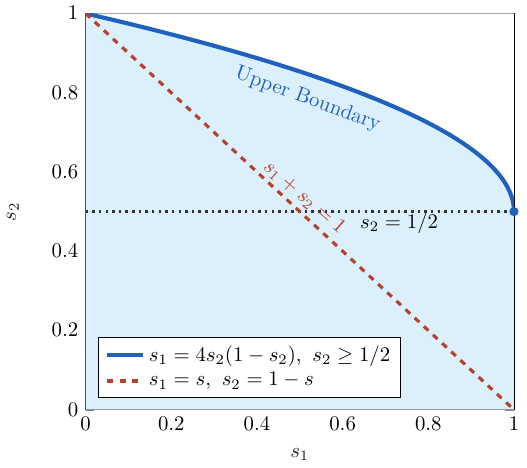}
    \caption{Illustration of the region $\mathcal{R}$ \eqref{eq:region}, shaded in the figure.}
    \label{fig:Region}
\end{figure}

Figure~\ref{fig:Region} shows the region $\mathcal{R}$ defined in \eqref{eq:region}. Setting $s_1 = s_2 = s$ in \eqref{eq:ineqnspq} yields the multiplicative bound
    \begin{equation}
        \operatorname{Tr}[\min\{A,B\}] \le \frac{1}{s} \sum_i \min(a_i, 1) w_i,
    \end{equation}
provided that $(s,s) \in \mathcal{R}$.
    The tightest of such bound is achieved at the intersection $s = 4s(1-s)$, which gives $s = 3/4$. This prefactor cannot be tightened further. To see this, consider the example
    \begin{equation}
        A = \biggl(\gamma - \frac{1}{2}\biggr)\ket{+}\bra{+} + \frac{1}{2} \ket{-}\bra{-}, \qquad B = \ket{0}\bra{0}. \label{eq:ABexample}
    \end{equation}
    In this case, it follows that
    \begin{equation}
        \lim_{\gamma \to \infty} \frac{\operatorname{Tr}[\min\{A,B\}]}{\sum_i \min(a_i,b_1)w_i} = \frac{4}{3}.
    \end{equation}
    This example is precisely the one considered in~\cite[Eq.~(67)]{cheng2026multiple}, where the authors say \textit{``It is, however, not clear if the factor $4/3$ is universal.''} 
    Lemma~\ref{lemma3} demonstrates that the limiting ratio of this example is not an isolated instance, and that, for any pure state $\rho$ and mixed state $\sigma$, the inequality
    \begin{equation}
        \operatorname{Tr}[\min\{t \rho,\sigma\}] \le \frac{4}{3} \sum_{i,j} \min(t P_{ij}, Q_{ij}), \quad t \geq 0,
    \end{equation}
    holds universally.

A direct application of the variational formulation from Lemma~\ref{lemma:1} together with Lemma~\ref{lemma3} yields an upper bound on the quantum error trade-off:
\begin{theorem}\label{thm_pure}
 Let $\rho,\sigma \in \mathcal{D}(\mathcal{H})$, and let $\rho$ be a pure state. Then, the quantum error trade-off $\beta_{\alpha}(\rho, \sigma)$ is upper-bounded as
\begin{equation}
    \beta_{\alpha} (\rho, \sigma) \leq \frac{1}{s_1}\beta_{s_2\alpha}(P, Q)
\end{equation}
for all $(s_1, s_2) \in \mathcal{R}$ defined in \eqref{eq:region}.
\end{theorem}
\begin{IEEEproof}
Fix $t \ge 0$. Invoking Lemma~3 with $A = \sigma$ and $B = t \rho$, and subtracting $t \alpha$ from both sides, we have
\begin{IEEEeqnarray}{lCl}
\Tr[\min\{t\rho, \sigma\}] - t\alpha &\le &\sum_{ij} \min\left(\frac{t P_{ij}}{s_2}, \frac{Q_{ij}}{s_1}\right) - t\alpha \nonumber\\
&= &\frac{1}{s_1} \biggl\{ \sum_{ij} \min\left(t' P_{ij}, Q_{ij}\right) - s_2 t'\alpha \biggr\}, \label{eq:lemma2_beta}
\end{IEEEeqnarray}
where $t' \triangleq (s_1 / s_2) t$. Since $s_1, s_2 > 0$, the mapping $t \mapsto t'$ is a bijection on $[0, \infty)$. Taking the supremum over $t \ge 0$ on both sides yields
\begin{IEEEeqnarray}{lCl}
\beta_{\alpha}(\rho, \sigma) &\le &\frac{1}{s_1} \sup_{t' \ge 0} \biggl\{ \sum_{ij} \min\left(t' P_{ij}, Q_{ij}\right) - s_2 t'\alpha \biggr\},
\end{IEEEeqnarray}
which, by \eqref{eq:var_class}, completes the proof.
\end{IEEEproof}

The feasible region $\mathcal{R}$, illustrated in Figure~\ref{fig:Region}, provides a clear geometric link between the bounds from Theorems~\ref{thm_general} and~\ref{thm_pure}. Indeed, the line segment given by $s_1 = s$ and $s_2 = 1-s$ for $s \in (0,1)$ is entirely contained within $\mathcal{R}$. Consequently, relaxing the optimal bound of Theorem~\ref{thm_pure} along this parameter choice yields a direct, alternative proof of Theorem~\ref{thm_general} for the pure--mixed regime. Moreover, optimizing the bound in Theorem 3 over the boundary of $\mathcal{R}$, i.e., $s_1 = 4 s_2 (1-s_2)$, yields tighter results in the pure-mixed regime.

\section{Asymptotic Analysis}
\label{sec:asymp}
In this section, we show that the lower and upper bounds in Theorems~\ref{thm:low} and \ref{thm_general}, with an adequate choice of the parameter $s$, provide a unified and conceptually simple framework for recovering asymptotic results in the small-deviations (second-order), moderate-deviations, and large-deviations regimes. To this end, we extend the hypothesis test from single copies $\rho, \sigma$ to $n$-fold tensor products $\rho^{\otimes n}, \sigma^{\otimes n}$.
Within this section, we assume $\supp(\rho) \subseteq \supp(\sigma)$ such that all the relevant quantities in \eqref{eq:in}-\eqref{eq:fin} are bounded. Finally, we define the hypothesis testing relative entropy
\begin{IEEEeqnarray}{lCl}
    D_h^{\varepsilon}(\rho^{\otimes n} \| \sigma^{\otimes n}) &\triangleq & 
    -\log\beta_{\varepsilon}\bigl(\rho^{\otimes n}, \sigma^{\otimes n}\bigr). \label{eq:hrel}
\end{IEEEeqnarray}
For the classical Nussbaum-Szkoła distributions $P$ and $Q$, $D_h^{\varepsilon}(P^n \| Q^n)$ is defined analogously.
\subsection{Small Deviations}
First, we examine the asymptotic behavior of the quantum error trade-off subject to a fixed constraint $\varepsilon$ on the type I error probability. This setup corresponds to the small-deviations regime, also referred to as second-order asymptotics. In the following corollary, we use Theorems~\ref{thm:low} and~\ref{thm_general} to establish a precise asymptotic expansion up to the second-order term.
\begin{corollary}\label{thm:small-deviations}
    Let $V(P \| Q) > 0$. Then,
    \begin{equation}
        D_h^{\varepsilon}(\rho^{\otimes n}  \|  \sigma^{\otimes n}) =  nD(\rho  \|  \sigma) + \sqrt{n V(\rho  \|  \sigma)} \Phi^{-1}\left( \varepsilon \right) + O(\log{n}), \label{eq:prop1}
    \end{equation}
    where $\Phi(x)$ denotes the cumulative distribution function of the standard normal distribution and $O(\log n)$ summarizes terms of order $\log n$.
\end{corollary}
\begin{IEEEproof}
Taking logarithms on both sides of \eqref{eq:thmlow}, and replacing $s$ by $s_n$ to indicate that this parameter may depend on $n$, we obtain from Theorem~\ref{thm:low} that
\begin{equation}
    D_h^{\varepsilon}(\rho^{\otimes n}  \|  \sigma^{\otimes n}) \leq D_h^{\frac{\varepsilon}{1-s_n}}(P^n  \|  Q^n) -\log{s_n}.\label{eq:thm1rel}
\end{equation}
From the small-deviations asymptotic behavior of the classical hypothesis testing problem \cite[Prop.~2.3]{tan2014asymptotic}, we have
\begin{equation}
    D_h^{\varepsilon}(P^n  \|  Q^n) 
    = nD(P  \|  Q) + \sqrt{n V(P  \|  Q)} \Phi^{-1}\left( \varepsilon \right)
    + \frac{1}{2}\log{n}  + O(1). \label{eq:prop23}
\end{equation}
Using \eqref{eq:prop23} in \eqref{eq:thm1rel}, and setting $s_n = 1/{\sqrt{n}}$, we obtain that
\begin{equation}
    D_h^{\varepsilon}(\rho^{\otimes n} \| \sigma^{\otimes n}) \leq nD(P \| Q) 
     + \sqrt{n V(P  \|  Q)} \Phi^{-1}\left( \frac{\varepsilon}{1-n^{-1/2}} \right)
    + \log{n} + O(1). \label{eq:relentuppsmall1}
\end{equation}
Performing a first-order Taylor expansion of $\Phi^{-1}( \frac{\varepsilon}{1-n^{-1/2}})$ around $\varepsilon$, and grouping terms of order $\log n$ and $1$ in one $O(\log n)$-term, proves the converse.

For the achievability, we use Theorem~\ref{thm_general} to obtain
\begin{equation}
    D_h^{\varepsilon}(\rho^{\otimes n}  \|  \sigma^{\otimes n}) \geq D_h^{s_n\varepsilon}(P^n  \|  Q^n) -\log{\frac{1}{1-s_n}}.\label{eq:thm_purerel}
\end{equation}
Again, from the classical small deviations result in \cite[Prop.~2.3]{tan2014asymptotic}, \eqref{eq:thm_purerel} becomes
\begin{equation}
D_h^{\varepsilon}(\rho^{\otimes n} \| \sigma^{\otimes n}) \geq nD(P \| Q) 
     + \sqrt{n V(P  \|  Q)} \Phi^{-1}\left( s_n \varepsilon \right)
    + \frac{1}{2}\log{n} +\log{\left( 1 - s_n \right)}+ O(1).
\end{equation}
Next, we choose $s_n = 1 - 1/\sqrt{n}$ and perform a first order Taylor expansion of $\Phi^{-1}(\cdot)$ around $\varepsilon$, i.e.,
\begin{IEEEeqnarray}{lCl}
    \Phi^{-1}\left(\frac{\sqrt{n}-1}{\sqrt{n}}\varepsilon \right)
    &= &\Phi^{-1} (\varepsilon) + O\left(\frac{1}{\sqrt{n}}\right),
\end{IEEEeqnarray}
where the terms in $O(1/\sqrt{n})$ are grouped into $O(1)$ after being multiplied by the prefactor $\sqrt{n}$ in front of the dispersion term in \eqref{eq:thm_purerel}. After substituting $s_n = 1 - 1/\sqrt{n}$ in $\log{(1-s_n)}$, \eqref{eq:thm_purerel} simplifies to 
\begin{equation}
    D_h^{\varepsilon} ( \rho^{\otimes n} \| \sigma^{\otimes n}) \geq nD(P  \|  Q) + \sqrt{n V(P  \|  Q)} \Phi^{-1}\left( \varepsilon \right) + O(1). \label{eq:smallachiev}
\end{equation}
Combining \eqref{eq:relentuppsmall1} and \eqref{eq:smallachiev}, and using \eqref{eq:in} and \eqref{eq:relvariance}, we obtain \eqref{eq:prop1}.
\end{IEEEproof}

\subsection{Moderate Deviations}
The moderate deviations regime considers the scenario where, rather than imposing a fixed constraint $\varepsilon$ on the type~I error probability, the error probability decays subexponentially with $n$ as $\varepsilon_n = e^{-na_n^2}$, where $\{a_n\}$ is a moderate sequence defined below. As before, Theorems~\ref{thm:low} and~\ref{thm_general} enable us to establish the exact asymptotic expansion for this regime up to the second-order term.

\begin{corollary}\label{thm:moderate-deviations}
Let $\{a_n\}$ be a nonnegative sequence satisfying $a_n \to 0$ and
$n a_n^2 \to\infty$ as $n\to\infty$. Further let $\varepsilon_n = e^{-n a_n^2}$. Then,
\begin{equation} \label{eq:moderate}
D_h^{\varepsilon_n}(\rho^{\otimes n} \| \sigma^{\otimes n}) = nD(\rho\|\sigma) - \sqrt{2 V(\rho\|\sigma)} n a_n + o(n a_n)
\end{equation}
where $o(na_n)$ summarizes terms that grow more slowly than $na_n$ as $n\to\infty$.
\end{corollary}
\begin{IEEEproof}
We begin by upper-bounding $D_h^{\varepsilon_n}(\rho^{\otimes n} \| \sigma^{\otimes n})$ using Theorem~1 and Lemma~1
\begin{IEEEeqnarray}{lCl}
D_h^{\varepsilon_n}(\rho^{\otimes n} \| \sigma^{\otimes n}) & \leq & - \log\left(Q^n\left(\log\frac{P^n}{Q^n} > R_n\right) + 2^{-R_n}\left( P^n\left(\log\frac{P^n}{Q^n} \leq R_n\right)-\frac{1}{1-s}\varepsilon_n\right)\right) - \log s \nonumber\\
& \leq & - \log\left(2^{-R_n}\left( P^n\left(\log\frac{P^n}{Q^n} \leq R_n\right)-\frac{1}{1-s}\varepsilon_n\right)\right) - \log s \nonumber\\
& = & R_n - \log\left(P^n\left(\log\frac{P^n}{Q^n} \leq R_n\right)-\frac{1}{1-s}\varepsilon_n\right) - \log s ,\label{eq:UB}
\end{IEEEeqnarray}
where the second step follows from the monotonicity of
the logarithm function. For arbitrary choices of $R_n$ and $s < 1/2$, we next set
\begin{equation}
\label{eq:R_n_def}
R_n = n D(P\|Q) - n a_n \lambda_n
\end{equation}
for some $\lambda_n$ that we shall define in the following.
Specifically, we will choose $\lambda_n$ such that
\begin{equation}
\label{eq:LB_P}
P^n\left(\log\frac{P^n}{Q^n} \leq R_n\right) \geq \frac{1}{1-2s} \varepsilon_n.
\end{equation}
In this case, the last two terms on the right-hand side of \eqref{eq:UB} are bounded by
\begin{IEEEeqnarray}{lCl}
 -\log\left(P^n\left(\log\frac{P^n}{Q^n} \leq R_n\right)-\frac{1}{1-s}\varepsilon_n\right) -\log s & \leq & -\log\left(\frac{1}{1-2s}\varepsilon_n - \frac{1}{1-s} \varepsilon_n\right) - \log s \nonumber\\
 & = & -\log \varepsilon_n + \log\frac{s}{(1-2s)(1-s)} - \log s \nonumber\\
 & = & n a_n^2 +O(1)\nonumber\\
 & = & o(n a_n) \label{eq:forget_logP}
\end{IEEEeqnarray}
where the last two step follow because $\varepsilon_n = e^{-n a_n^2}$ for $\{a_n\}$ satisfying $a_n\to 0$ and $n a_n^2\to\infty$, so $n a_n^2$ and $O(1)$ are both $o(n a_n)$ since $\frac{n a_n^2}{n a_n}=a_n \to 0$ and $n a_n = \sqrt{n}\sqrt{n a_n^2}\to\infty$.

We prove \eqref{eq:LB_P} using the following result by Rozovsky:
\begin{lemma}[\hspace{-0.01mm}{\cite[Th.~B2]{rozovsky2002estimate}}]\label{lem:rozovsky}\label{rozovsky}
Let $X_1,\ldots,X_n$ be independent random variables with finite third moments
\begin{equation}
\mu_k = \E{X_k}, \quad \sigma_k^2 = \Var{X_k}, \quad t_k = \E{|X_k-\mu_k|^3}.
\end{equation}
Let $S= \sigma_1^2+\ldots+\sigma_n^2$ and $T=t_1+\ldots+t_n$, and let $A_1>0$ and $A_2>0$ be universal constants. Then, for every $x\geq 1$, we have that
\begin{equation}
\Prob\left(\sum_{k=1}^n (X_k - \mu_k) > x \sqrt{S}\right) \geq Q(x) e^{-\frac{A_1 T}{S^{3/2}}x^3}\left(1-\frac{A_2 T}{S^{3/2}}x\right).
\end{equation}
\end{lemma}
To apply Lemma~\ref{rozovsky}, we set
\begin{IEEEeqnarray}{lCl}
X_k & = & -\log\frac{P}{Q}, \\
\mu_k & = & - D(P\|Q), \\
\sigma_k^2 & = & V(P\|Q), \\
t_k & = & T(P\|Q).
\end{IEEEeqnarray}
We can then write the probability on the left-hand side of \eqref{eq:LB_P} as
\begin{IEEEeqnarray}{lCl}
P^n\left(\log\frac{P^n}{Q^n} \leq R_n\right)
& = & \Prob\left(\sum_{k=1}^n (X_k-\mu_k) \geq -R_n - n \mu_k\right) \nonumber\\
& = &  \Prob\left(\sum_{k=1}^n (X_k-\mu_k) \geq \lambda_n n a_n\right)
\end{IEEEeqnarray}
where the last step follows from \eqref{eq:R_n_def}. We next argue that \eqref{eq:LB_P} follows by setting
\begin{equation}
\lambda_n^2 = 2 V(P\|Q)\left(1-\frac{B_n}{n a_n^2} - \frac{C_n}{n a_n^2}+ \frac{\log(1-2 s)}{n a_n^2}\right)
\end{equation}
where
\begin{IEEEeqnarray}{lCl}
B_n & \triangleq & \frac{\sqrt{8}A_1T(P\|Q)}{V^{3/2}(P\|Q)}n a_n^3 - \log\left(1-\frac{\sqrt{2}A_2 T(P\|Q)}{V(P\|Q)}a_n\right) \\
C_n & \triangleq & \frac{1}{2}\log(4\pi n a_n^2) - \log\left(1-\frac{1}{n a_n^2}\right).
\end{IEEEeqnarray}
For future reference, note that $B_n$ and $C_n$ are nonnegative and of order $o(n a_n^2)$. Consequently,
\begin{equation}
\label{eq:lambda_asymp}
\lambda_n^2 = 2 V(P\|Q) + o(1)
\end{equation}
which in turn implies that, for sufficiently large $n$, $V(P\|Q) \leq \lambda_n^2 \leq 2 V(P\|Q)$.

To show that the above $\lambda_n$ yields \eqref{eq:LB_P}, we use Lemma~\ref{rozovsky} together with $S = n V(P\|Q)$ and $T=n T(P\|Q)$ to lower-bound the probability on the left-hand side of \eqref{eq:LB_P} as (cf.~\cite[Eq.~(75)]{polyanskiy2010channel})
\begin{equation}
\label{eq:LB_p1}
P^n\left(\log\frac{P^n}{Q^n} \leq R_n\right) \geq Q\left(\frac{\lambda_n}{\sqrt{V(P\|Q)}}\sqrt{n a_n^2}\right) e^{- \frac{\lambda_n^3 A_1 T(P\|Q)}{V^3(P\|Q)}n a_n^3}\left(1-\frac{\lambda_n A_2 T(P\|Q)}{V^2(P\|Q)}a_n\right).
\end{equation}
Taking logarithms on both sides of \eqref{eq:LB_p1}, and using the lower bound $Q(x) \geq 1/\sqrt{2\pi x^2} e^{-x^2/2}(1-x^{-2})$, $x>0$, we then obtain that
\begin{IEEEeqnarray}{lCl}
\log P^n\left(\log\frac{P^n}{Q^n} \leq R_n\right) & \geq & \log Q\left(\frac{\lambda_n}{\sqrt{V(P\|Q)}}\sqrt{n a_n^2}\right) -  \frac{\lambda_n^3 A_1 T(P\|Q)}{V^3(P\|Q)}n a_n^3  + \log\left(1-\frac{\lambda_n A_2 T(P\|Q)}{V^2(P\|Q)}a_n\right)\nonumber\\
& \geq & \log Q\left(\frac{\lambda_n}{\sqrt{V(P\|Q)}}\sqrt{n a_n^2}\right) -  B_n \nonumber\\
& \geq & - \frac{\lambda_n^2}{2V(P\|Q)}  n a_n^2 - \frac{1}{2}\log\left(\frac{2\pi\lambda_n^2}{V(P\|Q)} n a_n^2\right) + \log\left(1-\frac{V(P\|Q)}{\lambda_n^2 n a_n^2}\right) -B_n \nonumber\\
& \geq & - \frac{\lambda_n^2}{2V(P\|Q)}  n a_n^2 - B_n - C_n\nonumber\\
& = & -n a_n^2 + B_n + C_n - \log(1-2 s) - B_n - C_n \nonumber\\
& = & -n a_n^2 - \log(1-2 s). \label{eq:LB_this}
\end{IEEEeqnarray}
Here, the second and fourth inequality follow by bounding $V(P\|Q) \leq \lambda_n^2 \leq 2 V(P\|Q)$; the third inequality follows from the lower bound on the $Q$ function. It follows that
\begin{equation}
P^n\left(\log\frac{P^n}{Q^n} \leq R_n\right) \geq e^{-n a_n^2-\log(1-2 s)} = \frac{1}{1-2 s} \varepsilon_n
\end{equation}
which is \eqref{eq:LB_P}.

Then, from \eqref{eq:UB}--\eqref{eq:forget_logP} and \eqref{eq:lambda_asymp},
\begin{IEEEeqnarray}{lCl}
D_h^{\varepsilon_n}(\rho^{\otimes n} \| \sigma^{\otimes n}) & \leq & n D(P\|Q) -n a_n \lambda_n + o(n a_n) \nonumber\\
& = & n D(P\|Q) - n a_n \sqrt{2 V(P\|Q)+o(1)} + o(n a_n) \nonumber\\
& = &  n D(P\|Q) - \sqrt{2 V(P\|Q)} n a_n + o(n a_n) \label{eq:last}
\end{IEEEeqnarray}
where the last step follows by performing a Taylor series of $\sqrt{2V(P\|Q)+o(1)}$ around $2 V(P\|Q)$. Together with \eqref{eq:in} and \eqref{eq:relvariance}, this proves that the right-hand side of \eqref{eq:moderate} is an upper bound on $D_h^{\varepsilon_n}(\rho^{\otimes n} \| \sigma^{\otimes n})$.

To derive a lower bound on $D_h^{\varepsilon_n}(\rho^{\otimes n} \| \sigma^{\otimes n})$, we begin by bounding the hypothesis testing relative entropy using the bound in Theorem~\ref{thm_general} for some fixed value $s$
\begin{equation}
    D_h^{\varepsilon_n}(\rho^{\otimes n} \| \sigma^{\otimes n}) \geq -\log\left( Q^n \left( \log\frac{P^n}{Q^n} > R_n\right) \right) +\log{s} \label{eq:achiev_mod_new}
\end{equation}
where $R_n$ is an arbitrary threshold satisfying
\begin{equation}
    P^n \left( \log\frac{P^n}{Q^n} \leq R_n\right) \leq (1-s)\varepsilon_n. \label{eq:Pnleq}
\end{equation}
Indeed, the Neyman-Pearson trade-off $\beta_\alpha(P^n,Q^n)$ is non-increasing in $\alpha$, so any threshold satisfying \eqref{eq:Pnleq} yields the required bound. We next bound the probability in \eqref{eq:achiev_mod_new} by using that, on the event $\{\log (P^n/Q^n) > R_n\}$, we have $Q^n < 2^{-R_n}P^n$. Therefore
\begin{equation}
    Q^n\left( \log\frac{P^n}{Q^n}  > R_n\right) \leq 2^{-R_n} P^n\left( \log\frac{P^n}{Q^n} > R_n \right). \label{eq:ineqprob2}
\end{equation}
Substituting \eqref{eq:ineqprob2} into \eqref{eq:achiev_mod_new} gives
\begin{IEEEeqnarray}{lCl}
     D_h^{\varepsilon_n}(\rho^{\otimes n} \| \sigma^{\otimes n}) &\geq & R_n -\log\left(1-P^n\left(\log\frac{P^n}{Q^n} \leq R_n\right)\right) +\log s. \label{eq:achiev_mod_new2}
\end{IEEEeqnarray}
We next choose 
\begin{equation}
    R_n = nD(P\|Q) + \lambda_n \label{eq:lambda_def}
\end{equation}
for a suitable choice $\lambda_n \leq 0$ and define the zero-mean random variable
\begin{equation}
    X_k = \log \frac{P}{Q}-D(P\|Q).
\end{equation}
The probability constraint \eqref{eq:Pnleq} then becomes
\begin{equation}
    P^n \left( -\sum_{k=1}^n X_k \geq -\lambda_n\right) \leq (1-s)\varepsilon_n. \label{eq:Pnleq2}
\end{equation}
The assumption $\supp(\rho) \subseteq \supp(\sigma)$ implies that $P \ll Q$, so on the support of $P$, $|X_k| \leq M$ for some finite $M$. Bernstein's inequality \cite[Corollary~2.11]{boucheron2003concentration} yields then
\begin{equation}
    P^n\left(-\sum_{k=1}^n X_k \geq -\lambda_n\right) \leq \exp\left(-\frac{\lambda_n^2}{n2V(P\|Q) - \frac{2}{3}M\lambda}\right).
\end{equation}
Hence, \eqref{eq:Pnleq} can be satisfied by choosing
\begin{equation}
    \frac{\lambda_n^2}{n2V(P\|Q) - \frac{2}{3}M\lambda} = na_n^2 - \log (1-s). \label{eq:berns_eq2}
\end{equation}
Solving the resulting quadratic equation and taking its negative root gives
\begin{equation}
    \lambda_n = -\left(\frac{M}{3} + \sqrt{\frac{M^2}{9} + \frac{n2V(P\|Q)}{na_n^2 - \log (1-s)}}\right)(na_n^2 - \log (1-s)). \label{eq:lambdan}
\end{equation} 
Grouping terms of order smaller than $n a_n$, \eqref{eq:lambdan} can be simplified to 
\begin{IEEEeqnarray}{lCl}
    \lambda_n &= & -\sqrt{\frac{M^2}{9} + \frac{n2V(P\|Q)}{na_n^2 - \log (1-s)}}na_n^2 + o(na_n) \nonumber\\
    &= & -\sqrt{\frac{n2V(P\|Q)}{na_n^2 - \log (1-s)} + O(1)} na_n^2 +o(na_n)\nonumber\\
    &=& -\sqrt{\frac{n^3a_n^42V(P\|Q)}{na_n^2 - \log (1-s)}} + o(na_n). \label{eq:taylor1}
\end{IEEEeqnarray}
Finally, a first-order Taylor expansion of \eqref{eq:taylor1} around $n^2a_n^22V(P\|Q)$ yields
\begin{IEEEeqnarray}{lCl}
    \sqrt{\frac{n^3a_n^42V(P\|Q)}{na_n^2 - \log (1-s)}} &= &\sqrt{n^2a_n^22V(P\|Q) + \frac{n^2a_n^22V(P\|Q)\log s}{na_n^2 - \log (1-s)}}\nonumber\\
    &=& na_n\sqrt{2V(P\|Q) + o(1)} \nonumber\\
    &=&\sqrt{2V(P\|Q)}na_n + o(na_n),
\end{IEEEeqnarray}
so
\begin{equation}
    \lambda_n = -\sqrt{2V(P\|Q)}na_n + o(na_n). \label{eq:taylor2}
\end{equation}

Combining \eqref{eq:Pnleq}, \eqref{eq:lambda_def}, and \eqref{eq:taylor2} with \eqref{eq:achiev_mod_new2}, we obtain that
\begin{IEEEeqnarray}{lCl}
D_h^{\varepsilon_n}(\rho^{\otimes n} \| \sigma^{\otimes n}) & \geq & n D(P \| Q) - \sqrt{2V(P\|Q)}na_n - \log\bigl(1-(1-s) \varepsilon_n\bigr) + \log s + o(n a_n) \nonumber\\
& = & n D(P \| Q) - \sqrt{2V(P\|Q)}na_n + o(n a_n)
\end{IEEEeqnarray}
since $\varepsilon_n$ and $\log s$ are both of order $o(n a_n)$. We thus obtain that the right-hand side of \eqref{eq:moderate} is also a lower bound on $D_h^{\varepsilon_n}(\rho^{\otimes n} \| \sigma^{\otimes n})$.
\end{IEEEproof}

\subsection{Large Deviations}
Finally, we consider the large deviations regime, where the type I error probability decays exponentially with $n$ as $\varepsilon_n = e^{-nr}$ at a fixed rate $r > 0$. In this setting, Theorems~\ref{thm:low} and~\ref{thm_general} provide a direct framework for characterizing the optimal error exponent.

\begin{corollary}\label{thm:large}
    Let $\varepsilon_n = e^{-n r}$ where $0 < r < D(\rho \| \sigma)$. Then,
\begin{IEEEeqnarray}{lCl}
  \mathop{\lim}_{n\to\infty} \frac{1}{n} D_h^{\varepsilon_n}(\rho^{\otimes n} \| \sigma^{\otimes n}) & = & \sup_{0< p < 1} \left\{ D_p(\rho\|\sigma) + \frac{p}{p-1} r \right\}\label{eq:hoeffding-bound}
\end{IEEEeqnarray}
where $D_p(\rho\|\sigma)$ was defined in \eqref{eq:D_p}.
\end{corollary}
\begin{IEEEproof}
We first prove the converse by taking Theorem~1 and setting $\varepsilon_n = e^{-nr}$, where $0 < r< D(P \| Q)$. For some value $s$, this yields
\begin{equation}
    D_h^{\varepsilon_n}(\rho^{\otimes n} \| \sigma^{\otimes n}) \leq -\log s -\log\beta_{\frac{1}{1-s}\varepsilon_n}(P^n, Q^n).
\end{equation}
Observe that
\begin{equation}
    \lim_{n \to \infty} \frac{1}{n} \log\left( \frac{1}{1-s} \varepsilon_n\right) = -r.
\end{equation}
It follows from the classical large deviations result \cite[Theorem~1]{CsiszarLongo1971} that
\begin{IEEEeqnarray}{lCl}
    \limsup_{n \to \infty} \frac{1}{n} D_h^{\varepsilon_n}(\rho^{\otimes n} \| \sigma^{\otimes n}) &\leq & -\lim_{n \to \infty} \frac{1}{n}\log \beta_{\frac{1}{1-s} \varepsilon_n}(P^n, Q^n) \nonumber\\
    &= &\sup_{0\leq p\leq 1} \left\{ D_p(\rho\|\sigma) + \frac{p}{p-1} r \right\}. \label{eq:LD_UB}
\end{IEEEeqnarray}

For the achievability, we apply the upper bound in Theorem~\ref{thm_general} with some fixed $s$ to obtain
\begin{equation}
    D_h^{\varepsilon_n}(\rho^{\otimes n} \| \sigma^{\otimes n}) \geq -\log\frac{1}{1-s} - \log \beta_{s \varepsilon_n} (P^n, Q^n).
\end{equation}
Similar to the converse, we have that 
\begin{equation}
    \lim_{n\to \infty} \frac{1}{n}\log(s \varepsilon_n) = -r.
\end{equation}
We then obtain again by \cite[Theorem~1]{CsiszarLongo1971} that
\begin{IEEEeqnarray}{lCl}
    \liminf_{n \to \infty} \frac{1}{n} D_h^{\varepsilon_n}(\rho^{\otimes n} \| \sigma^{\otimes n}) &\geq& -\lim_{n \to \infty} \frac{1}{n} \log\beta_{s \varepsilon_n}(P^n, Q^n) \nonumber\\
    & = & \sup_{0 < p < 1} \left\{ D_p(\rho\|\sigma) + \frac{p}{p-1} r \right\}. \label{eq:LD_LB}
\end{IEEEeqnarray}
Combining \eqref{eq:LD_UB} and \eqref{eq:LD_LB} proves the corollary.
\end{IEEEproof}

\section{Numerical Results and Discussion}
\label{sec:discussion}
In this section, we compare the converse and achievability bounds in Theorems~\ref{thm:low}--\ref{thm_pure} with previous bounds obtained in the literature in the non-asymptotic setting.
\subsection{Comparison with Existing Converse Bounds}

To evaluate the converse bound established in Theorem~\ref{thm:low}, we compare it against four benchmark results from the literature:

\begin{enumerate}
    \item The harmonic-mean converse bound established by Cheng et al.~\cite[Th.~4.2]{cheng2026multiple}:
    \begin{equation}
        \beta_{\alpha}(\rho, \sigma) \geq \sup_{t\geq0} \left\{\text{HM}(t\rho, \sigma)-t\alpha\right\}, \label{eq:hmbound}
    \end{equation}
    where, for the spectral decomposition in \eqref{eq:spec}, the Petz-Nussbaum-Szkoła harmonic mean is defined as
    \begin{equation}
        \text{HM}(t\rho, \sigma)  = \sum_{ij}\frac{t\lambda_i\mu_j}{t\lambda_i + \mu_j}\lvert\braket{x_i | y_j}\rvert^2.
    \end{equation}
    \item The Nussbaum--Szkoła bound on the minimum average error probability from Audenaert et al.~\cite[Prop.~2]{audenaert2008}, which, when extended to the asymmetric setting, coincides with Theorem~\ref{thm:low} for $s = 1/2$
    \begin{equation}
        \beta_\alpha(\rho, \sigma) \geq \frac{1}{2}\beta_{2\alpha}(P,Q).
    \end{equation}
    \item The fidelity-based converse bound from Pereira et al.~\cite[Eq.~(10)]{pereira2023analytical}, given by
    \begin{equation}
        \beta_{\alpha}(\rho, \sigma) \geq \alpha - 2\alpha F^2(\rho, \sigma) + F^2(\rho, \sigma) - 2 \sqrt{(1-\alpha)\alpha(1-F^2(\rho, \sigma))F^2(\rho, \sigma)}, \label{eq:pure_fidelity}
    \end{equation}
    where the fidelity is defined as
    \begin{equation}
        F(\rho, \sigma) \triangleq \Tr\left[\left|\sqrt{\rho}\sqrt{\sigma} \right|\right] =\Tr\left[\sqrt{\sqrt{\rho}\sigma \sqrt{\rho}}\right].
    \end{equation}
    \item The information-spectrum converse bound from Tomamichel and Hayashi~\cite[Eq.~(27)]{tomamichel2013hierarchy}, given by
    \begin{equation}
        \beta_{\alpha}(\rho, \sigma)  \geq \frac{\delta^{4}(1-\alpha-\delta)}{2^{8}(\alpha+\delta)\theta}2^{-D_{s}^{\alpha+\delta}(P\|Q)}\label{eq:TomHay}
    \end{equation}
    for $0 < \delta \leq \min\{\alpha, 1 - \alpha\}$, where $P$ and $Q$ are the Nussbaum--Szkoła distributions, $\theta$ is the minimum between the number of distinct eigenvalues of $\rho$ and $\sigma$, and the information-spectrum relative entropy is defined as
    \begin{equation}
        D_s^{\alpha}(P \| Q) = \sup \left\{ R \in \mathbb{R} \; \middle| \; \Tr[\rho \{\rho \preceq 2^R\sigma \}] \leq \alpha \right\}.
    \end{equation}
\end{enumerate}

\begin{figure}[tb]
    \centering
    \includegraphics[width=0.48\linewidth]{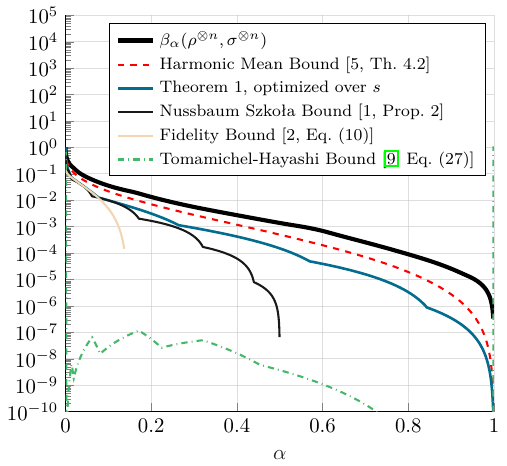}
    \caption{Comparison of converse bounds for the $5$-fold tensor products $\rho^{\otimes 5}$ and $\sigma^{\otimes 5}$ of the states in \eqref{eq:mixedstates}.}
    \label{fig:HTmixed}
\end{figure}

Figure~\ref{fig:HTmixed} compares the resulting converse bounds with the true quantum error trade-off function $\beta_\alpha(\rho^{\otimes n}, \sigma^{\otimes n})$
for $n= 5$ and the mixed states
\begin{equation}
\rho = \begin{bmatrix} 0.5 & 0.4 \\ 0.4 & 0.5\end{bmatrix}, \qquad \sigma =\begin{bmatrix} 0.9 & 0 \\ 0 & 0.1\end{bmatrix}. \label{eq:mixedstates}
\end{equation}
Observe the limited range of the Nussbaum--Szkoła bound in \cite[Prop.~2]{nussbaum2006lower} and the fidelity bound in \cite[Eq.~(10)]{pereira2023analytical},
which yield finite non-trivial bounds only for $0 \leq \alpha < F(\rho, \sigma)^{2n}$ and $0 \leq \varepsilon < 1/2$, respectively.
In contrast, the bound of Theorem~\ref{thm:low}, optimized over $s\in[0,1]$, 
provides a tight characterization of the error trade-off over the entire range $0 \leq \alpha < 1$. Furthermore, the Tomamichel--Hayashi bound \cite[Eq.~(27)]{tomamichel2013hierarchy} yields finite values over $0 < \alpha < 1$, but it is significantly looser than the bound obtained from Theorem~\ref{thm:low} for the considered example. We further note that, while the bound in \cite[Eq.~(27)]{tomamichel2013hierarchy} recovers the small- and moderate-deviation results of Corollaries~\ref{thm:small-deviations} and \ref{thm:moderate-deviations}
(see~\cite{tomamichel2013hierarchy} and \cite{chubb2017moderate}),
it fails to recover the error exponent in the large-deviations analysis beyond Stein's regime.

The tightest bound is the harmonic mean bound presented in \cite[Th. 4.2]{cheng2026multiple}. In contrast to the bound in Theorem~\ref{thm:low}, it manages to also achieve the correct third-order term in the small deviations regime; see \cite[Remark 6.8]{cheng2026multiple}. However, it is more difficult to evaluate analytically, and it must be relaxed in order to obtain the large deviations result. We believe the bound presented in Theorem 1 strikes a good balance between tightness and analytic tractability.

\subsection{Comparison with Existing Achievability Bounds}
To evaluate the performance of the bounds in Theorems~\ref{thm_general} and \ref{thm_pure}, we compare them against established achievability results in the literature across two scenarios: mixed-mixed states and pure-mixed states. We compare these bounds against the following benchmarks:
\begin{enumerate}
    \item The second-order achievability bound from Tomamichel and Hayashi~\cite[Eq.~28]{tomamichel2013hierarchy}, given by
    \begin{equation}
        \beta_{\alpha}(\rho, \sigma) \leq \frac{\delta}{\theta} 2^{-D_{s}^{\alpha-\delta}(P\|Q)}, \label{eq:tomhayAch}
    \end{equation}
    where $\delta = \min \{\varepsilon, 1 - \varepsilon\}$.

        \item The Chernoff-overlap-based upper bound \cite[Prop. 3.2]{audenaert2012quantum}\footnote{The bound is in a slightly different form as in \cite[Prop. 3.2]{audenaert2012quantum}. Specifically, \eqref{eq:audenaert} corresponds to ~\cite[ Eq.~(17)]{pereira2023analytical}, where the Lagrange multiplier $t$ in \cite[Prop. 3.2]{audenaert2012quantum} is already optimized.}
    \begin{equation}\label{eq:audenaert}
        \beta_{\alpha}(\rho,\sigma) \leq \inf_{0 < p < 1} \Biggl\{(1-p) Q_p^{1/(1-p)} \left(\frac{p}{\alpha}\right)^{p/(1-p)}\Biggr\},
    \end{equation}
    where
    \begin{equation}\label{eq:qp}
        Q_p \triangleq e^{\kappa(p)} = \Tr\left[\rho^p\sigma^{1-p}\right].
    \end{equation}
    \item The harmonic-mean achievability bound from Cheng et al.~\cite[Th.~4.2]{cheng2026multiple}
    \begin{equation}
        \beta_\alpha (\rho, \sigma) \leq \inf_{0 < s < 1} \Biggl\{ \frac{1}{1-s}\sup_{t \geq 0} \bigl \{ \text{HM}(t\rho, \sigma) -t s \alpha \bigr \}\Biggr\}, \label{eq:thm42ach}
    \end{equation}
    
\end{enumerate}

\subsubsection{Mixed-Mixed States}
First, we consider hypothesis testing between two mixed states. Figure~\ref{fig:HTmixed_low} compares the achievability bounds against the true quantum error trade-off function for $n=5$ and the states \eqref{eq:mixedstates}.
\begin{figure}[tb]
    \centering
    \includegraphics[width=0.48\linewidth]{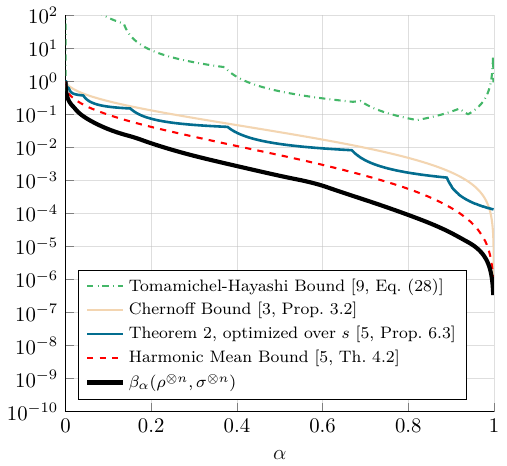}
    \caption{Comparison of achievability bounds for the $5$-fold tensor products $\rho^{\otimes 5}$ and $\sigma^{\otimes 5}$ of the states \eqref{eq:mixedstates}.}
    \label{fig:HTmixed_low}
\end{figure}
As expected in the mixed-mixed scenario, the harmonic-mean bound \cite[Th.~4.2]{cheng2026multiple} provides the tightest approximation of the error trade-off over the entire range $0 \leq \alpha < 1$. This occurs because the bound in Theorem~\ref{thm_general} is a relaxation of \cite[Th.~4.2]{cheng2026multiple}. Similar to the converse bound, the harmonic-mean achievability bound recovers the exact third-order term in the small deviations regime, but it must be relaxed to obtain the large-deviations result. Again, we believe Theorem~\ref{thm_general} strikes a good balance between tightness and analytic tractability.

The Chernoff-distance-based bound \cite[Eqs.~(19)--(20)]{pereira2023analytical} provides looser approximations compared to \cite[Th.~4.2]{cheng2026multiple} and Theorem~\ref{thm_general}. While it cannot recover small- and moderate-deviations results, its construction via $Q_p$ allows for a straightforward recovery of large-deviations achievability results. Moreover, because $Q_p$ factorizes under tensor products, its numerical evaluation is significantly less demanding than the evaluation of the other bounds.

\subsubsection{Pure-Mixed States}

Next, we consider the scenario where $\rho$ is a pure state and $\sigma$ is a mixed state, specifically using the states
\begin{equation}
    \rho = \begin{bmatrix}
        0.5 & 0.5 \\ 0.5 & 0.5
    \end{bmatrix}, \qquad 
    \sigma = \begin{bmatrix} 
        0.9 & 0 \\ 0 & 0.1
    \end{bmatrix}. \label{eq:purestates}
\end{equation}

We compare \cite[Eq.~(28)]{tomamichel2013hierarchy} and \cite[Th.~4.2]{cheng2026multiple} against Theorem~\ref{thm_pure} optimized over $(s_1,s_2)\in \mathcal{R}$ by setting $s_1 = 4s(1-s)$ and $s_2=s$, i.e.,
\begin{equation}
    \beta_{\alpha} (\rho, \sigma) \leq \frac{1}{4 s (1-s)} \beta_{s \alpha}(P, Q), \label{eq:betalphaupper}
\end{equation}
optimized over $s \in (0,1)$. 

\begin{figure}[tb]
    \centering
    \includegraphics[width=0.48\linewidth]{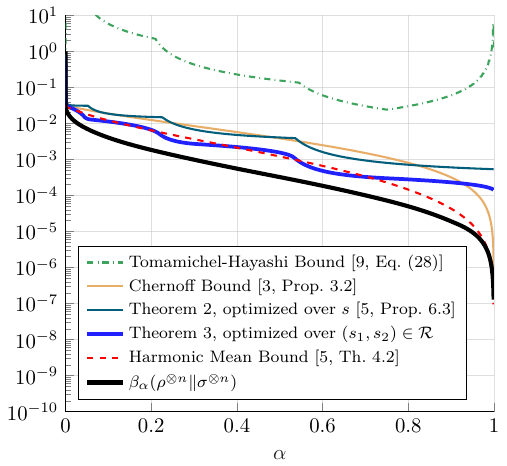}
    \caption{Comparison of achievability bounds for the $5$-fold tensor products $\rho^{\otimes 5}$ and $\sigma^{\otimes 5}$ of the states  \eqref{eq:purestates}.}
    \label{fig:HTpure}
\end{figure}

Figure~\ref{fig:HTpure} displays relative performance trends among the common achievability bounds consistent with those in Figure~\ref{fig:HTmixed_low}.
In this pure-mixed scenario, we also include Theorem~\ref{thm_pure} optimized over $(s_1,s_2)\in \mathcal{R}$ which is tighter than Theorem~\ref{thm_general} in this setting.
Observe that Theorem~\ref{thm_pure} exhibits a similar tightness as the harmonic-mean bound. In fact, neither of the two bounds dominates over the entire range $0 \leq \alpha \leq 1$.

\appendices
\section{Proof of Lemma~3}\label{appendixA}
In this appendix, we establish a rigorous connection between the quantum operator minimum $\Tr[\min\{A,B\}]$ for non-commuting operators and its classical counterpart under the Nussbaum--Szkoła mapping \eqref{eq:Ns-mapping}. Specifically, we consider an arbitrary operator $A \succeq 0$ and a rank-1 operator $B \succeq 0$ with spectral decompositions
\begin{equation}\label{eqn:QHT-AB-def}
    A = \sum_i a_i \ket{x_i}\!\bra{x_i}, \qquad B = b_1 \ket{y_1}\!\bra{y_1}.
\end{equation}
We associate with the pair $(A,B)$ the commuting operators
\begin{equation}
\label{eq:Lambda-Gamma}
\begin{aligned}
P_{A,B}
\triangleq \bigoplus_{i,j}
a_i \lvert\braket{x_i|y_j}\rvert^2,\qquad
Q_{A,B}
\triangleq \bigoplus_{i,j}
b_j \lvert\braket{x_i|y_j}\rvert^2,
\end{aligned}
\end{equation}
acting on $\mathcal H\otimes\mathcal H$.
Before presenting the proof of Lemma~\ref{lemma3}, we establish the necessary notation and two auxiliary results (a lemma and a proposition) that form the core of the derivation. First, recall from \eqref{eq:trmin} that the quantum operator minimum can be rewritten as
\begin{equation}
    \Tr[\min\{A,B\}] = \Tr[B] - \Tr\!\bigl[(A-B)_-\bigr], \label{eq:trmin2} 
\end{equation}
where $(A-B)_-$ denotes the negative part of the self-adjoint operator $A-B$. For positive weights $s_1, s_2 > 0$ and the overlap coefficients $w_i$ defined in \eqref{eq:weights_i}, the trace of the classical minimum evaluates directly to
\begin{equation}
    \Tr\biggl[\min\biggl\{\frac{P_{A,B}}{s_1}, \frac{Q_{A,B}}{s_2}\biggr\}\biggr] = \sum_i w_i \min\!\left(\frac{a_i}{s_1},\,\frac{b_1}{s_2}\right).
\end{equation}
To evaluate the quantum trace term $\Tr[(A-B)_-]$ in \eqref{eq:trmin2}, we first analyze the spectral properties of the rank-1 perturbation $A-B$.
\begin{lemma}\label{lem:rank1-facts}
Let $A$ and $B$ be as defined in \eqref{eqn:QHT-AB-def}.
The operator $A-B$ has at most one negative eigenvalue. Moreover, if
$\lambda$ is an eigenvalue of $A-B$ and $\lambda\neq a_i$ for every $i$, then
$\lambda$ is a root of the secular equation
\begin{equation}\label{eqn:rank1-secular-equation}
1-b_1\sum_i \frac{w_i}{a_i-\lambda}=0.
\end{equation}
\end{lemma}

\begin{IEEEproof}
For the second statement, apply the matrix
determinant lemma to obtain
\begin{equation}
\det(A-B-\lambda I)
=
\det(A-\lambda I)
\left(
1-b_1\bra{y_1}(A-\lambda I)^{-1}\ket{y_1}
\right),
\end{equation}
and expand the expectation $\bra{y_1}(A-\lambda I)^{-1}\ket{y_1}$ in the eigenbasis of $A$. The first statement, i.e., the fact that there is at most one negative eigenvalue follows from Weyl's interlacing theorem specialized to rank-1 perturbations \cite[Corollary~4.3.9]{horn2012matrix}, which states that, if $\lambda_i(\cdot)$ denotes the $i$-th largest eigenvalue, then $\lambda_{i+1}(A) \leq \lambda_i(A-B) \leq \lambda_{i}(A)$.
\end{IEEEproof}

\begin{proposition}\label{thm:rank1-two-param}
Let $A\succeq0$ and let $B$ be rank one. For $0<s_1,s_2\leq1$,
\begin{equation}\label{eqn:rank1-ratio-bound}
\min_{A,B} \left\{\frac{\Tr\bigl[\min\bigl\{\frac{P_{A,B}}{s_1}, \frac{Q_{A,B}}{s_2}\bigr\}\bigr]}{\Tr\bigl[\min\{A,B\}\bigr]}\right\}
=
R^\star (s_1,s_2), 
\end{equation}
where
\begin{equation}\label{eqn:R-s1-s2-star}
R^\star (s_1,s_2)
\triangleq
\begin{cases}
\displaystyle \frac{4s_2-s_1}{4s_2^2}, & 0<s_1\leq 2s_2,\\[6pt]
\displaystyle \frac{1}{s_1}, & 2s_2<s_1\leq1.
\end{cases}
\end{equation}
\end{proposition}

\begin{IEEEproof}
The ratio is invariant under common scaling of $A$ and $B$, so we normalize $b_1=1$. The numerator on the left-hand side of \eqref{eqn:rank1-ratio-bound} becomes
\begin{equation}
\Tr\biggl[\min\biggl\{\frac{P_{A,B}}{s_1}, \frac{Q_{A,B}}{s_2}\biggr\}\biggr] = \sum_i
w_i \min\!\left(\frac{a_i}{s_1},\,\frac{1}{s_2}\right).
\end{equation}
By \eqref{eq:trmin2} and Lemma~\ref{lem:rank1-facts}, if
$\lambda$ is the smallest eigenvalue of $A-B$, then
\begin{equation}\label{eqn:rank1-Pq-two-param}
\Tr[\min \{A, B\}]
=1+\min(\lambda,0).
\end{equation}
Then, the minimization on the left-hand side of \eqref{eqn:rank1-ratio-bound} subject to the secular constraint \eqref{eqn:rank1-secular-equation} can be written as
\begin{equation}\label{eq:miniratio}
\begin{aligned}
\underset{\bm{w}, \bm{a}, \lambda}{\text{minimize}} \quad & \frac{ \sum_{i=1}^{d} w_i \min\left( \frac{a_i}{s_1}, \frac{1}{s_2} \right)}{1 + \min(\lambda, 0)} \\[6pt]
\text{subject to} \quad & \sum_{i=1}^d \frac{w_i}{a_i - \lambda} = 1,
\end{aligned}
\end{equation}
where the optimization is carried out over the domain of $a_i \ge 0$ and probability vectors $\bm{w}$ satisfying $\sum_{i=1}^d w_i = 1$. Let
\begin{equation}
\mathcal{S}_1 \triangleq \left\{i:a_i<\frac{s_1}{s_2}\right\},
\qquad
\mathcal{S}_2 \triangleq \left\{i:a_i\geq\frac{s_1}{s_2}\right\}.
\end{equation}
Then, the objective of \eqref{eq:miniratio} becomes
\begin{equation}\label{eqn:R-s1-s2-expanded}
f(\bm{w}, \bm{a}, \lambda) = \frac{
\frac{1}{s_1}\sum_{i\in \mathcal{S}_1} w_i a_i
+
\frac{1}{s_2}\sum_{i\in \mathcal{S}_2}w_i
}
{1+\min(\lambda,0)}.
\end{equation}
For fixed $\lambda$, form the Lagrangian
\begin{IEEEeqnarray}{lCl}
\mathcal L(\bm{w}, \bm{a},\lambda,\nu)
&=&
\frac{
\frac{1}{s_1}\sum_{i\in \mathcal{S}_1} w_i a_i
+
\frac{1}{s_2}\sum_{i\in \mathcal{S}_2} w_i
}
{1+\min(\lambda,0)}
+
\nu\left(\sum_i\frac{w_i}{a_i-\lambda}-1\right),
\end{IEEEeqnarray}
where $\nu$ is the Lagrange multiplier associated with the constraint in \eqref{eq:miniratio}.

The stationarity equations imply that all $a_i$ with $i\in \mathcal{S}_1$ are equal, and we denote the common value by $a$. For $i\in \mathcal{S}_2$, the objective is independent of
$a_i$ while the constraint contribution decreases as $a_i$ grows, so the
optimum pushes these variables to infinity. Hence, with
\begin{equation}
W_1 \triangleq \sum_{i\in \mathcal{S}_1}w_i,
\end{equation}
the active constraint reduces to
\begin{equation}
\frac{W_1}{a-\lambda}=1,
\qquad\text{or equivalently}\qquad
a=\lambda+W_1.
\end{equation}
The feasibility condition $\bm{a}\geq0$ gives $\lambda\geq-W_1$. The reduced objective is
\begin{equation}\label{eqn:f-s1-s2}
f(\lambda,W_1)
=
\frac{
\frac{W_1}{s_1}(\lambda+W_1)
+
\frac{1-W_1}{s_2}
}
{1+\min(\lambda,0)},
\qquad
0\leq W_1\leq1,\quad \lambda\geq-W_1.
\end{equation}

\medskip\noindent
We next minimize $f(\lambda,W_1)$ over $\lambda$. If $\lambda\geq0$, then
\begin{equation}
f(\lambda,W_1)
=
\frac{W_1}{s_1}\lambda
+
\frac{W_1^2}{s_1}
+
\frac{1-W_1}{s_2},
\end{equation}
which is minimized at $\lambda=0$. For $-W_1 \le \lambda < 0$ and $W_1 > 0$, the objective function $f(\lambda, W_1)$ takes the form
\begin{equation}
f(\lambda, W_1) = \frac{W_1}{s_1} \frac{\lambda + g(W_1)}{\lambda + 1},
\label{eq:objective_lambda_W1}
\end{equation}
where 
\begin{equation}
g(W_1) \triangleq W_1 + \frac{s_1}{s_2} \frac{1 - W_1}{W_1}.
\end{equation}
Differentiating with respect to $\lambda$ yields
\begin{equation}
\frac{\partial f}{\partial \lambda} = \frac{W_1}{s_1} \frac{1 - g(W_1)}{(\lambda + 1)^2}.
\end{equation}
The sign of $\frac{\partial f}{\partial \lambda}$ is determined entirely by $1 - g(W_1)$, partitioning the minimization over $[-W_1, 0)$ into two distinct regimes. When $g(W_1) < 1$, $f(\lambda, W_1)$ is strictly increasing, and the minimum is attained at the lower boundary $\lambda = -W_1$, yielding
\begin{equation}
f(-W_1, W_1) = \frac{\frac{1-W_1}{s_2}}{1-W_1} = \frac{1}{s_2},
\end{equation}
which is independent of $W_1$. Nevertheless, we next show that considering $g(W_1) \geq 1$ yields a smaller objective, and therefore the global minimum is attained in this region:

When $g(W_1) \ge 1$, $f(\lambda, W_1)$ is non-increasing in $\lambda$, so its minimum is attained at the upper boundary $\lambda = 0$,
\begin{equation}
f(0,W_1)
=
\frac{W_1^2}{s_1}
+
\frac{1-W_1}{s_2}. \label{eq:objective_fin}
\end{equation}
It remains to minimize \eqref{eq:objective_fin} over $0\leq W_1 \leq 1$. The unconstrained minimizer is
$W_1^\star=\frac{s_1}{2s_2}$.
Projecting onto $[0,1]$ gives
\begin{equation}
W_1^\star
=
\begin{cases}
\displaystyle \frac{s_1}{2s_2}, & 0<s_1\leq2s_2,\\[6pt]
1, & 2s_2<s_1\leq1.
\end{cases}\label{eq:W_opt}
\end{equation}
Note that $g(W_1^\star) \ge 1$ holds for both cases. For $W_1^\star = \frac{s_1}{2s_2}$, we have $g\left(\frac{s_1}{2s_2}\right) = 2 - \frac{s_1}{2s_2} \ge 1$ precisely when $s_1 \le 2s_2$, while for $W_1^\star = 1$, $g(1) = 1$ holds trivially. This validates that $W_1^\star$ is indeed the optimal choice of $W_1$ for $\lambda = 0$.

Evaluating $f(0, W_1^\star)$ across both parameter ranges yields then
\begin{equation}
    f\left(0, \frac{s_1}{2s_2}\right) = \frac{1}{s_2} - \frac{s_1}{4s_2^2} = \frac{4s_2 - s_1}{4s_2^2} < \frac{1}{s_2} \qquad 0 < s_1 \leq s_2,
\end{equation}
and
\begin{equation}
    f(0, 1) = \frac{1}{s_1} < \frac{1}{2s_2} < \frac{1}{s_2}, \qquad 2s_2 < s_1 < 1.
\end{equation}
Thus, $f(0, W_1^\star) < f(-W_1, W_1) = \frac{1}{s_2}$ strictly holds in both ranges, confirming that the choice $(\lambda, W_1) =  (0, W_1^\star)$ globally minimizes the error trade-off \eqref{eq:objective_lambda_W1}. Substituting \eqref{eq:W_opt} into \eqref{eq:objective_fin} directly yields $R^\star (s_1,s_2)$ as defined in \eqref{eqn:R-s1-s2-star}.
\end{IEEEproof}

With Proposition~\ref{thm:rank1-two-param} established, the proof of Lemma~3 reduces to characterizing the parameter set $(s_1, s_2)$ for which $R^\star (s_1,s_2) \ge 1$. We analyze the two branches of \eqref{eqn:R-s1-s2-star} separately:

\begin{itemize}
    \item $0 < s_1 \le 2s_2$: Here, $R^\star (s_1,s_2) = \frac{4s_2 - s_1}{4s_2^2}$. The condition $R^\star (s_1,s_2) \ge 1$ is equivalent to
    \begin{equation}
        4s_2 - s_1 \ge 4s_2^2 \iff s_1 \le 4s_2(1 - s_2).
    \end{equation}
    When $0 \le s_2 \le \frac{1}{2}$, $4s_2(1 - s_2) \le 2s_2$, so the constraint $s_1 \le 4s_2(1 - s_2)$ is tighter than $s_1 \le 2s_2$ and defines the valid boundary. When $\frac{1}{2} < s_2 < 1$, we have $4s_2(1 - s_2) > 2s_2$, so $s_1 \le 4s_2(1 - s_2)$ is automatically satisfied for all $0 < s_1 \le 2s_2$.

    \item $2s_2 < s_1 \le 1$: Here, $R^\star (s_1,s_2) = \frac{1}{s_1}$. The requirement $R^\star (s_1,s_2) \ge 1 \iff s_1 \le 1$ holds trivially across the entire branch. This branch is non-empty whenever $2s_2 < 1$, i.e., $0 \le s_2 < \frac{1}{2}$.
\end{itemize}

Combining these two branches confirms that $R^\star (s_1,s_2) \ge 1$ if, and only if, $(s_1, s_2) \in \mathcal{R}$, as defined in \eqref{eq:region}. Consequently, for any $(s_1,s_2) \in \mathcal{R}$, Proposition~\ref{thm:rank1-two-param} yields
\begin{equation}
\frac{\Tr\bigl[\min\bigl\{\frac{P_{A,B}}{s_1}, \frac{Q_{A,B}}{s_2}\bigr\}\bigr]}{\Tr\bigl[\min\{A,B\}\bigr]} \geq 1,
\end{equation}
which completes the proof of Lemma~3.

\section*{Acknowledgments}
Stimulating discussions with Hao-Chung Cheng are gratefully acknowledged. 

OpenAI Codex was used to assist in developing the proof of Theorem 3. The authors verified and revised the mathematical arguments.

\bibliographystyle{IEEEtran}
\bibliography{bib/references}

\end{document}